\documentclass[prl, aps, floatfix, twocolumn, preprintnumbers, superscriptaddress]{revtex4-1}
\usepackage{hyperref}
\usepackage{verbatim}
\usepackage{amsmath}
\usepackage{latexsym}

\usepackage{natbib}
\usepackage{amsfonts}
\usepackage{amsmath}
\usepackage{amssymb}
\usepackage{amsthm}
\usepackage{graphicx}
\usepackage{bm}

\newcommand{\be}{\begin{eqnarray} \begin{aligned}}
\newcommand{\ee}{\end{aligned} \end{eqnarray} }
\newcommand{\benn}{\begin{eqnarray*} \begin{aligned}}
\newcommand{\eenn}{\end{aligned} \end{eqnarray*} }

\newcommand{\bc}{\begin{center}}
\newcommand{\ec}{\end{center}}
\newcommand{\half}{\frac{1}{2}}

\newcommand{\id}{\mathbb{I}}

\newcommand{\tr}{\mathop{\mathrm{tr}}\nolimits}

\newcommand{\iu}{i}

\newtheorem{theorem}{Theorem}[section]

\newtheorem{lemma}[theorem]{Lemma}
\newcommand{\ra}{\rightarrow}
\newcommand{\hil}{\mathcal{H}}

\newcommand{\cI}{\ensuremath{\mathcal{I}}}

\newcommand{\nn}{\nonumber}

\usepackage{amsfonts}
\def\Real{\mathbb{R}}

\def\id{\mathbb{I}}

\def\01{\{0,1\}}

\newcommand{\eps}{\varepsilon}
\newcommand{\ket}[1]{|#1\rangle}
\newcommand{\bra}[1]{\langle#1|}

\newcommand{\proj}[1]{|#1\rangle\langle#1|}

\newcommand{\ignore}[1]{}

\newenvironment{sdp}[2]{
\smallskip
\begin{center}
\begin{tabular}{ll}
#1 & #2\\
subject to
}
{
\end{tabular}
\end{center}
\smallskip
}

\newcommand{\hmin}{{\ensuremath{{\rm H}}_{\min}}}

\newcommand{\ent}{{\ensuremath{ { \rm H} } }}
\newcommand{\sys}{{\ensuremath{ S } }}
\newcommand{\env}{{\ensuremath{ E } }}
\newcommand{\ssys}{_{\sys}}
\newcommand{\senv}{_{\env}}
\newcommand{\HSE}{H_{\sys\env}}
\newcommand{\rhoSE}{\rho_{\sys\env}}
\newcommand{\rhoS}{\rho\ssys}
\newcommand{\dS}{d\ssys}
\newcommand{\dE}{d\senv}

\newcommand{\fmS}{\frac{\id\ssys}{\dS}}
\newcommand{\cS}{\ensuremath{\mathcal{S}}}
\newcommand{\cT}{\ensuremath{\mathcal{T}}}
\newcommand{\mdag}{^{\dag}}
\newcommand{\lmax}{\lambda_{\max}}
\newcommand{\cj}{Choi-Jamio\l{}kowski }
\DeclareMathOperator{\Herm}{Herm}
\DeclareMathOperator{\End}{End}

\begin{document}

\title{Almost all quantum states have low entropy rates for any coupling to
the environment}
\author{Adrian Hutter}
\affiliation{Centre for Quantum Technologies, National University of Singapore, 2 Science Drive 3, 117543 Singapore}
\author{Stephanie Wehner}
\affiliation{Centre for Quantum Technologies, National University of Singapore, 2 Science Drive 3, 117543 Singapore}

\begin{abstract}
	The joint state of a system that is in contact with an environment is called \emph{lazy}, if the entropy rate
	of the system under \emph{any} coupling to the environment is zero. Necessary and sufficient conditions 
	have recently been established for a state to be lazy [Phys. Rev. Lett. 106, 050403 (2011)], and it was
	shown that almost all states of the system and the environment do \emph{not}
	have this property [Phys. Rev. A 81, 052318 (2010)]. At first glance, this may lead us to believe
	that low entropy rates themselves form an exception, in the sense that most states are far from being lazy and have high 
	entropy rates.
	Here, we show that in fact the opposite is true if the environment
	is sufficiently large. Almost all 
	states of the system and the environment are \emph{pretty lazy} -- their entropy rates are low for \emph{any}
	coupling to the environment.
\end{abstract}
\maketitle

A central question in the study of decoherence and thermalization is how the entropy of a system $\sys$ changes over time when 
it is in contact with an environment $\env$~\cite{lazyPaper}. The entropy of the system $\sys$ is thereby typically measured in terms of the von Neumann
entropy $\ent(\sys) = - \tr(\rho_{\sys} \log \rho_{\sys})$, and quantifies the degree of decoherence of the system~\cite{log_comment}.
Two extreme cases help to illustrate
this measure: If we initially prepare the system in a known pure state, then its entropy is $\ent(\sys)=0$ -- no decoherence has yet taken place.
However, if the system becomes fully mixed later on all information about its initial state is lost, and at this point
its entropy scales with its dimension $\ent(\sys) = \log d_{\sys}$.
To determine the rate of decoherence, i.e.\ ``information loss'' over time 
one is interested in the so-called \emph{entropy rate}~\cite{lazyPaper}
\begin{align}
	\frac{d\ent(\sys)}{dt}\ ,
\end{align}
of the system evolving according to a coupling Hamiltonian $\HSE$
\begin{align}
	\rhoSE(t) = \exp(-i\HSE t)\rhoSE(0)\exp(i\HSE t)\ .
\end{align}
Since the von Neumann entropy $\ent(\sys)$ also measures the degree of entanglement between the system and the environment,
we can equally well think of this quantity as a measure of the rate at which a particular interaction can create entanglement 
between the system and its environment. Indeed, the value of this derivative at time $t=0$ is more commonly known in the quantum information community 
as the \emph{entangling rate} of a particular coupling Hamiltonian $\HSE$~\cite{duer:defRate,childs:rate2,sergey:bestRate}.

How large can this entangling rate be? Intuitively, it is clear that this rate should depend on the interaction
strength between the system and the environment. Note that we can write any coupling Hamiltonian as
\begin{align}\label{eq:decomposition}
	\HSE = c \id_{\sys\env} + H_{\sys} \otimes \id_{\env} + \id_{\sys} \otimes H_{\env} + H_{\rm int}\ ,
\end{align}
where $c$ is a constant.
Since the non-interacting terms $H_{\sys} \otimes \id_{\env}$ and $\id_{\sys} \otimes H_{\env}$ do not contribute to the creation of entanglement between the system
and the environment, the interaction strength is often measured in terms of $\|H_{\rm int}\|_\infty$. 
That is, in terms of the largest eigenvalue of $H_{\rm int}$. A more involved notion of the interaction strength will be introduced later on in the paper.
Following~\cite{duer:defRate,childs:entanglement,kraus:separability,wang:entanglement,childs:rate2,bennett:capacities}, it 
has been shown~\cite{sergey:bestRate} that for any pure state $\rho_{\sys\env}$ and
interaction Hamiltonian $H_{\sys\env}$ we have
\begin{align}\label{eq:rateBound}
	\left|\frac{d\ent(\sys)}{dt}\right| \leq c' \|H_{\rm int}\|_{\infty} \log d_{\sys}\ , 
\end{align}
where $c'$ is a constant.
For completeness sake, we provide a simple proof for $c'=4$ in the appendix.
This bound is essentially optimal, as it was shown that for any $d_{\sys} \leq d_{\env}$ there
exists a state with a very large entropy rate. That is, there exists an interaction Hamiltonian $\HSE$
such that its entropy rate is $O(\|H_{\rm int}\|_{\infty} \log d_{\sys})$, scaling with the dimension of the system $d_{\sys}$.

Are there many states with such high entropy rates? Recent work~\cite{lazyPaper} tackled the 
problem of studying entropy rates from the other end by providing necessary and 
sufficient conditions for a state $\rhoSE$ to have \emph{zero} entropy rate for \emph{any} Hamiltonian $\HSE$ at time 
$t=0$~\cite{time_comment}. Such states are also known as \emph{lazy states}.
In particular, it was shown that a state $\rhoSE$ is lazy if and only if
\begin{align}\label{eq:lazyCondition}
	[\rhoSE,\rho_{\sys} \otimes \id_{\env}] = 0\ .
\end{align}
Lazy states do not have to be eigenstates of $\HSE$ or $H_{\rm int}$, and have several properties
that are of interest when it comes to suppressing decoherence. In particular, it was suggested that
for a lazy state the entropy of the system could in principle be preserved by fast measurements
or dynamical decoupling techniques~\cite{facchi:subspaces,zanardi:symmetrizing,viola:decoupling}. 

Yet, lazy states are very unusual. In particular, it was shown~\cite{lazyPaper,cesarTalk} 
using the results of~\cite{ferraro:discord} that 
almost no states are lazy, in the sense that they have measure zero on the joint Hilbert space $\hil_{\sys} \otimes \hil_{\env}$ 
of the system and the environment~\cite{measure_comment}. At first glance, this may lead us to believe that low entropy rates themselves
are unusual, and that most states should have high entropy rates for at least some coupling Hamiltonian $\HSE$. 

\section{Result}

Here, we show that in fact the opposite is true if the environment is sufficiently large. Almost all states of the system and the environment are ``pretty lazy'', that is the entropy rate
on the system is very low for any coupling Hamiltonian. With low we thereby mean that the entropy rate scales as some vanishing parameter $\eps$ times the interaction strength.
Note that in contrast to the study of zero entropy rates, 
this is all one could hope for when talking about low entropy rates -- a stronger interaction strength will necessarily increase any non-zero rate. 

Our main result that almost all states have low entropy rates can now be stated slightly more formally.
In particular, we will show that the probability that a randomly chosen state $\rho_{SE}$ has large entropy rate is very small. 
That is,
\begin{align}\label{eq:mainResult}
	\Pr_{\rho_{SE}}\left[\left|\frac{d}{dt} H(S)_{\rho}\right| \geq \|H_{\rm int}\|_{\infty} \eps\right] \leq \delta\ ,
\end{align}
where 
\begin{align}
	\eps = 2^{-\frac{1}{2}\left(\log d_{\env} - 3 \log d_{\sys} - 4\right)},\ \delta = 2e^{-d_{\sys}^2/16}\ , 
\end{align}
and the distribution over the set of possible states on $\hil_S \otimes \hil_E$ can be any unitarily invariant measure.
If the environment is sufficiently large ($\log d_{\env} > 3 \log d_{\sys}$) and the system itself is not too small ($\log d_S > 2$), then
we obtain a strong statement. We will furthermore show a similar bound that is also interesting for extremely small systems $\log d_S \leq 2$
as long as $\log d_{\env} > (9/2) \log d_{\sys}$. In this case, we have 
\begin{align}
\eps = 2^{-\frac{1}{2} \left(\log {d_\env} - \frac{9}{2}\log d_{\sys} - 5\right)},\ \delta = 2e^{-d_\sys d^{1/3}_\env/16}\ .
\end{align}

Since the Hilbert space dimension grows exponentially with the number of constituent particles of a physical system and since we usually assume the environment $\env$ to consist of a large number of particles, at least one of the dimensional constraints will be fulfilled in typical situations of physical interest.

It is important to note that while the entropy rate in general depends on the relation between the Hamiltonian and the state (see \eqref{eq:vonNeumanRate}), the condition for a state being lazy expressed in \eqref{eq:lazyCondition} describes a property of the state alone. Similarly, given the discussed dimensional constraints, the very structure of most states $\rho_{SE}$ is such that they do not allow a fast change of the entropy in $S$ -- even for ``unphysical'' Hamiltonians $H_{SE}$.

In the appendix, we show that analogous results can be obtained for the linear entropy or purity, which has been studied in the context of entropy rates
in ~\cite{lazyPaper,hayashi:purity,gogolin:mthesis}. In this case, we even obtain slightly more favourable parameters.

\section{Proof}

Let us now see how we can prove said results. Our proof thereby proceeds in two steps. First of all, we recall 
that for a randomly chosen pure state from $\hil_\sys \otimes \hil_\env$ the state will almost certainly be close to
fully mixed on $\hil_\sys$, if the environment is significantly larger than the system~\cite{popescu:entanglement}. For completeness, we provide a simpler proof of this claim in the appendix. Second, we show that if a state is close to fully mixed on the system $\hil_\sys$ then it is indeed pretty lazy.

\smallskip
\noindent
{\bf Fully mixed on $\hil_\sys$:\ } 
Let us first consider only pure states on $\hil_\sys \otimes \hil_\env$. Note that chosing a random pure state according to the Haar measure is equivalent to applying a randomly chosen unitary $U$ to a fixed starting state, say, $\ket{0}_{\sys\env}$.
In contrast to~\cite{popescu:entanglement} our proof (see appendix) that such a random pure state is fully mixed on the system follows by an easy application of the decoupling theorem~\cite{dupuis:diss,dupuis:decoupling}. Furthermore, if we apply the decoupling theorem we do not have to restrict to pure states as in~\cite{popescu:entanglement}. That is, our statement does not only hold for most states of the form $U\proj{0}_{\sys\env}U\mdag$ but more generally for most states of the form $U\sigma_{\sys\env}U\mdag$ where $\sigma_{\sys\env}$ is an arbitrary state (pure or mixed) on $\hil_\sys \otimes \hil_\env$. Equivalently we may state that most states $\rhoSE$ with given eigenvalues and randomly chosen eigenstates are close to fully mixed on the system. ``Randomly chosen'' here means that the eigenbasis of $\rhoSE$ is chosen from the Haar measure, which by definition is unitarily invariant. Since our assertion holds for any fixed set of eigenvalues, it also holds if we pick $\rhoSE$ from any unitarily invariant measure on $\cS\left(\hil_\sys \otimes \hil_\env\right)$, the set of density operators on $\hil_\sys \otimes \hil_\env$.
Summarizing, we obtain the following little lemma, which is proven in the appendix.\\
\begin{lemma}\label{lem:closetoMixed}
For a bipartite system $\hil_\sys \otimes \hil_\env$ 
\begin{align}
\Pr_{\rho_{\sys\env}} \left\{\left\|\rho_{\sys}-\fmS\right\|_1 \geq \chi\right\} 
\leq \delta ,\label{eq:prob}
\end{align}
where the probability is computed over the choice of $\rhoSE$ from any unitarily invariant measure on $\cS\left(\hil_\sys \otimes \hil_\env\right)$,
and where we may choose either
\begin{align}\label{eq:parameters1}
\chi = 2^{-\half\left(\log \dE - \log \dS - 4\right)} , \delta = 2e^{-d_{\sys}^2/16}
\end{align}
or 
\begin{align}\label{eq:parameters2}
\chi = 2^{-\frac{1}{3}\left(\log \dE - \frac{3}{2} \log \dS - 5\right)} , \delta = 2e^{-\dS\dE^{1/3}/16}\ .
\end{align}
\end{lemma}

\smallskip
\noindent
{\bf Pretty lazy for the von Neumann entropy:\ }
Let us now turn to the main part of our proof. A small calculation~\cite{lazyPaper} shows that the rate of change of 
the von Neumann entropy is given by
\begin{align}\label{eq:vonNeumanRate}
	\frac{d\ent(\sys)}{dt} = -i \tr\left(H_{\rm int}\left[\log(\rho_\sys(t)) \otimes \id_\env,\rhoSE(t) \right]\right)\ .
\end{align}
Note that $[\log(\rho_{\sys})\otimes \id_\env,\rhoSE] = 0$ if and only if~\eqref{eq:lazyCondition} holds, and thus
the latter is a sufficient condition for a state $\rho_{\sys\env}$ to be lazy~\cite{lazyPaper}. Consider now a state $\rhoSE$ such that 
its reduced state $\rho_\sys = \tr_\env(\rhoSE) = \id_\sys/d_\sys$ is fully mixed. Clearly, any such state satisfies~\eqref{eq:lazyCondition}
and is a lazy state. 

How about states which are merely close to being fully mixed on $\hil_{\sys}$? The following lemma captures our
intuition that states which are close to lazy states on $\hil_{\sys}$ are in fact pretty lazy themselves.
Closeness it thereby measured in terms of the trace distance~\cite{nielsen&chuang:qc} 
which is the relevant quantity for distinguishing to quantum states~\cite{helstrom:detection}.

\begin{lemma}\label{lem:closeLazy}
	Consider a Hamiltonian with interaction strength $\left\|H_{\rm int}\right\|_\infty$. 
	For any quantum state $\rhoSE$ on $\HSE$ such that its reduced state is $\chi$-close to fully mixed, i.e.,
	$\chi = \|\rho_\sys- \id_\sys/d_\sys\|_1$ where $\chi \leq 1/d_\sys$ with $d_\sys \geq 2$,  
	its entropy rate is bounded by
	\begin{align}
		\left|\frac{d\ent(\sys)}{dt}\right| \leq \left\|H_{\rm int}\right\|_\infty  2 d_\sys \chi\ .
	\end{align}
\end{lemma}
\begin{proof}
	Using~\eqref{eq:vonNeumanRate} we can upper bound the entropy rate by
\begin{align}
&\left|\frac{d\ent(\sys)}{dt}\right| 
\leq \left\|H_{\rm int}\right\|_\infty\left\|\left[\log(\rho_S)\otimes\id_E,\rho_{SE}\right]\right\|_1 \label{eq:step1Bound}
\\&= \left\|H_{\rm int}\right\|_\infty\left\|\left[\left(\log(\rho_S)-\log(\fmS)\right)\otimes\id_E,\rho_{SE}\right]\right\|_1 \label{eq:step2Bound}
\\&\leq 2\left\|H_{\rm int}\right\|_\infty\left\|\left(\log(\rho_S)-\log(\fmS)\right)\otimes\id_E\right\|_\infty\left\|\rho_{SE}\right\|_1 \label{eq:step3Bound}
\\&= 2\left\|H_{\rm int}\right\|_\infty\left\|\log(\rho_S)-\log(\fmS)\right\|_\infty\ , \label{eq:stepLastBound}
\end{align}
where~\eqref{eq:step1Bound} follows from the fact that for any bounded operators $A$ and $B$
\begin{align}\label{normineq}
\left|\tr(AB)\right| \leq \tr\left|AB\right| = \left\|AB\right\|_1 \leq \left\|A\right\|_1 \left\|B\right\|_\infty\ \ ,
\end{align}
\eqref{eq:step3Bound} follows from the convexity of the L1-norm, and~\eqref{eq:stepLastBound} follows from the definition
of the L1-norm $\|A\|_1 = \tr\sqrt{A^\dagger A}$.
Now let $\left\{p_i\right\}_{i=1}^{d_S}$ denote the eigenvalues of $\rho_S$, so 
\begin{align}\label{wantmax}
\left|\frac{d\ent(\sys)}{dt}\right|\leq2\left\|H_{\rm int}\right\|_\infty\cdot\max_{i=1}^{d_S}\left|\log(p_id_S)\right|\ . 
\end{align}
We want to maximize the r.h.s.\ of (\ref{wantmax}) for fixed 
\begin{align}\label{chidef}
\chi=\left\|\rho_S - \fmS\right\|_1=\sum_{i=1}^{d_S}\left|p_i-\frac{1}{d_S}\right|\ .
\end{align} 
Without loss of generality, let $p_1$ denote the smallest eigenvalue and $p_2$ the largest, so $p_1 \leq \frac{1}{d_S} \leq p_2$. The quantity $\left|\log(p_id_S)\right|$ in \eqref{wantmax} is monotously decreasing in $p_i$ if $0\leq p_i \leq \frac{1}{d_S}$ and monotously increasing if $\frac{1}{d_S} \leq p_i \leq 1$. The following procedure therefore allows to increase the r.h.s.\ of (\ref{wantmax}) while keeping $\chi$ constant: For all $3 \leq i \leq d_S$, if $p_i < \frac{1}{d_S}$ replace $p_1 \mapsto p_1 + p_i - \frac{1}{d_S}$ and $p_i \mapsto \frac{1}{d_S}$. For all $3 \leq i \leq d_S$, if $p_i > \frac{1}{d_S}$ replace $p_2 \mapsto p_2 + p_i - \frac{1}{d_S}$ and $p_i \mapsto \frac{1}{d_S}$. We end up with $p_1=\frac{1}{d_S}-\frac{\chi}{2}$, $p_2=\frac{1}{d_S}+\frac{\chi}{2}$, $p_i=\frac{1}{d_S}$ for $3\leq i\leq d_S$. 
For $\chi\geq0$ we have 
\begin{align}
\left|\log(p_1d_S)\right|\geq\left|\log(p_2d_S)\right|
\end{align}
so that
\begin{align}
\left|\frac{d\ent(\sys)}{dt}\right|
&\leq2\left\|H_{\rm int}\right\|_\infty\cdot\left|\log(p_1d_S)\right|
\\&=2\left\|H_{\rm int}\right\|_\infty\cdot\left|\log\left(\left(\frac{1}{d_S}-\frac{\chi}{2}\right)d_S\right)\right|
\\&=2\left\|H_{\rm int}\right\|_\infty\cdot\left(-\log\left(1-\half d_S\chi\right)\right)\ .\label{eq:boundEnd}
\end{align}
Let us now upper bound the term on the r.h.s. Note that for $0 \leq x \leq \half$ the function $f(x):=-\log\left(1-x\right)$ is well defined
and convex.
By convexity we thus have $f(x)\leq2f(\half)x=2 x$ on the interval, and hence for $x = (1/2)d_\sys\chi \leq 1/2$ we have
\begin{align}\label{eq:chiUpper}
-\log\left(1-\half d_\sys\chi\right) \leq d_\sys \chi\ .
\end{align}
Upper bounding~\eqref{eq:boundEnd} using~\eqref{eq:chiUpper} now leads to the claimed result.
\end{proof}

Our claim that almost all states are pretty lazy now follows immediately by combining the two lemmas.
Lemma~\ref{lem:closetoMixed} tells us that the probability that a randomly chosen state $\rhoSE$ is $\chi$-close to maximally mixed on $\hil_\sys$ 
is extremely high, where $\chi = 2\sqrt{d_\sys/d_\env}$ and $\chi = 2\sqrt{d_\sys}/\sqrt[3]{d_\env}$ respectively.
Lemma~\ref{lem:closeLazy} now tells us that for sufficiently large $d_\env$ such states are indeed pretty lazy. The values for $\eps$ in~\eqref{eq:mainResult} are $d_\sys \chi$.

{\bf Interaction strengths.}
For completeness, we discuss how our bounds can be improved by a more refined measure of interaction strength.
First of all, note that the operators $H_{\sys}$ and $H_{\env}$ in~\eqref{eq:decomposition} are not unique and freedom in their choice may be used to minimize $\|H_{\rm int}\|_\infty$. Finding the optimal operators $H_{\sys}$ and $H_{\env}$ forms in general a non-trivial endeavour. Here, we do not deal with this general task but ask how to best define the ``interaction strength'' for a given operator $H_{\rm int}$. Indeed, the decomposition~\eqref{eq:decomposition} does become unique 
if we require
$H_{\sys}$ and $H_{\env}$ to be traceless and $H_{\rm int}$ to have vanishing partial traces on both $\sys$ and $\env$ (e.g.\ \cite{linden:fluctuations,lazyPaper}).
We can now measure the interaction strength
as
\begin{align}\label{eq:deltaDef}
	\Delta(H_{\rm int}) := 2 \min_{\lambda \in \Real}\|H_{\rm int} - \lambda \id_{SE}\|_{\infty}\ .
\end{align}
First of all, note that shifting all energy levels of a certain system by a constant amount does not affect the dynamics of that system. These only depend on the \emph{differences} between the energy eigenvalues. The quantity $H_{\rm int}$ as defined in the decomposition \eqref{eq:decomposition} is indeed invariant under addition of a multiple of $\id_{\sys\env}$ to $\HSE$. Similarly, we can see from \eqref{eq:vonNeumanRate} that adding a multiple of $\id_{\sys\env}$ to $H_{\rm int}$ alone does not affect the rate of change of the local entropy. For this reason, the quantity $\Delta(H_{\rm int})$ defined in~\eqref{eq:deltaDef}
provides a more robust measure of the ``interaction strength'' of $\HSE$ than $\left\|H_{\rm int}\right\|_{\infty}$, as it is already invariant under a shift of eigenvalues in $H_{\rm int}$. 
From~\eqref{eq:deltaDef} we furthermore see that this quantity can easily be computed using a semidefinite program (SDP)~\cite{boyd:convex} since we may also write
$\Delta(H_{\rm int}) = 2\gamma$ where $\gamma$ is the solution of the following SDP
\begin{sdp}{minimize}{$\gamma$}
	&$\gamma \id \geq H_{\rm int} - \lambda \id \geq -\gamma \id\ ,$
\end{sdp}
where the minimization is taken over variables $\gamma$ and $\lambda$.
Since $\Delta(H_{\rm int})$ equals the difference between the smallest and largest eigenvalue of $H_{\rm int}$ we have $\Delta(H_{\rm int}) \leq 2 \left\|H_{\rm int}\right\|_{\infty}$. An upper bound on the entangling rate which is proportional to $\left\|H_{\rm int}\right\|_{\infty}$ may therefore be strengthened by noting that we may replace $H_{\rm int}$ by $H_{\rm int} - \lambda\id$ without affecting time scales. This allows us to replace $\left\|H_{\rm int}\right\|_{\infty}$ by $\half\Delta(H_{\rm int})$ in all the bounds if desired.

\section{Discussion}

We have shown that almost all states of the system and the environment are in fact pretty lazy. If the environment $E$ is sufficiently larger than our system $S$ -- which we assume to be the case in physical scenarios -- the vast majority of bipartite states is such that their entropy in $S$ can only be changed at a vanishing rate. The relevant timescale is thereby given by the inverse of the interaction strength $\left\|H_{\rm int}\right\|_\infty$.
Our results should be compared to~\cite{linden:evolution,linden:fluctuations}  in which it
was shown that 
\emph{equilibration} is a generic property of pure states on $\hil_\sys \otimes \hil_\env$ if \env is sufficiently larger than $\sys$. That is, under this conditions almost all joint initial states will lead to the state of \sys being close to its temporal average for most times. Furthermore it is shown in \cite{linden:fluctuations} that for almost all joint initial states, the rate of change of $\sys$ (the speed of the fluctuations around the temporal average, that is)
will on average be small. The time scale that the speed of fluctuations is compared to is here given by $\left\|H\ssys\otimes\id_E + H_{\rm int}\right\|_\infty$. While only $H_{\rm int}$ is able to create entanglement between \sys and $\env$, both $H_{\rm int}$ and $H\ssys\otimes\id_E$ are relevant for the evolution of the state of $\sys$. If the rate of change of the state of \sys is low, this implies by Fannes' inequality~\cite{nielsen&chuang:qc} that the rate of change of the von Neumann entropy is low as well. So while the results of \cite{linden:evolution,linden:fluctuations} imply that most initial states will lead to entropy rates on \sys which in a long-time temporal average are low, we show that most bipartite states really \emph{are} such that the entropy rates on \sys are low for any interaction Hamiltonian.

\acknowledgments

We thank Cesar Rodriguez-Rosario for an inspiring talk at CQT, and interesting discussions as well as comments on our draft. This research was supported by the National Research Foundation and the Ministry of Education, Singapore.


\begin{thebibliography}{23}
\expandafter\ifx\csname natexlab\endcsname\relax\def\natexlab#1{#1}\fi
\expandafter\ifx\csname bibnamefont\endcsname\relax
  \def\bibnamefont#1{#1}\fi
\expandafter\ifx\csname bibfnamefont\endcsname\relax
  \def\bibfnamefont#1{#1}\fi
\expandafter\ifx\csname citenamefont\endcsname\relax
  \def\citenamefont#1{#1}\fi
\expandafter\ifx\csname url\endcsname\relax
  \def\url#1{\texttt{#1}}\fi
\expandafter\ifx\csname urlprefix\endcsname\relax\def\urlprefix{URL }\fi
\providecommand{\bibinfo}[2]{#2}
\providecommand{\eprint}[2][]{\url{#2}}

\bibitem[{\citenamefont{Rodr\'\i{}guez-Rosario
  et~al.}(2011)\citenamefont{Rodr\'\i{}guez-Rosario, Kimura, Imai, and
  Aspuru-Guzik}}]{lazyPaper}
\bibinfo{author}{\bibfnamefont{C.~A.} \bibnamefont{Rodr\'\i{}guez-Rosario}},
  \bibinfo{author}{\bibfnamefont{G.}~\bibnamefont{Kimura}},
  \bibinfo{author}{\bibfnamefont{H.}~\bibnamefont{Imai}}, \bibnamefont{and}
  \bibinfo{author}{\bibfnamefont{A.}~\bibnamefont{Aspuru-Guzik}},
  \bibinfo{journal}{Phys. Rev. Lett.} \textbf{\bibinfo{volume}{106}},
  \bibinfo{pages}{050403} (\bibinfo{year}{2011}).

\bibitem{log_comment}
All logarithms in this paper are base $2$.

\bibitem[{\citenamefont{D\"ur et~al.}(2001)\citenamefont{D\"ur, Vidal, Cirac,
  Linden, and Popescu}}]{duer:defRate}
\bibinfo{author}{\bibfnamefont{W.}~\bibnamefont{D\"ur}},
  \bibinfo{author}{\bibfnamefont{G.}~\bibnamefont{Vidal}},
  \bibinfo{author}{\bibfnamefont{J.~I.} \bibnamefont{Cirac}},
  \bibinfo{author}{\bibfnamefont{N.}~\bibnamefont{Linden}}, \bibnamefont{and}
  \bibinfo{author}{\bibfnamefont{S.}~\bibnamefont{Popescu}},
  \bibinfo{journal}{Phys. Rev. Lett.} \textbf{\bibinfo{volume}{87}},
  \bibinfo{pages}{137901} (\bibinfo{year}{2001}).

\bibitem[{\citenamefont{Childs et~al.}(2004)\citenamefont{Childs, Leung, and
  Vidal}}]{childs:rate2}
\bibinfo{author}{\bibfnamefont{A.}~\bibnamefont{Childs}},
  \bibinfo{author}{\bibfnamefont{D.}~\bibnamefont{Leung}}, \bibnamefont{and}
  \bibinfo{author}{\bibfnamefont{G.}~\bibnamefont{Vidal}},
  \bibinfo{journal}{Information Theory, IEEE Transactions on}
  \textbf{\bibinfo{volume}{50}}, \bibinfo{pages}{1189 } (\bibinfo{year}{2004}).

\bibitem[{\citenamefont{Bravyi}(2007)}]{sergey:bestRate}
\bibinfo{author}{\bibfnamefont{S.}~\bibnamefont{Bravyi}},
  \bibinfo{journal}{Phys. Rev. A} \textbf{\bibinfo{volume}{76}},
  \bibinfo{pages}{052319} (\bibinfo{year}{2007}).

\bibitem[{\citenamefont{Childs et~al.}(2003)\citenamefont{Childs, Leung,
  Verstraete, and Vidal}}]{childs:entanglement}
\bibinfo{author}{\bibfnamefont{A.}~\bibnamefont{Childs}},
  \bibinfo{author}{\bibfnamefont{D.}~\bibnamefont{Leung}},
  \bibinfo{author}{\bibfnamefont{F.}~\bibnamefont{Verstraete}},
  \bibnamefont{and} \bibinfo{author}{\bibfnamefont{G.}~\bibnamefont{Vidal}},
  \bibinfo{journal}{QIC} \textbf{\bibinfo{volume}{3}}, \bibinfo{pages}{097}
  (\bibinfo{year}{2003}).

\bibitem[{\citenamefont{Kraus et~al.}(2000)\citenamefont{Kraus, Cirac, Karnas,
  and Lewenstein}}]{kraus:separability}
\bibinfo{author}{\bibfnamefont{B.}~\bibnamefont{Kraus}},
  \bibinfo{author}{\bibfnamefont{J.~I.} \bibnamefont{Cirac}},
  \bibinfo{author}{\bibfnamefont{S.}~\bibnamefont{Karnas}}, \bibnamefont{and}
  \bibinfo{author}{\bibfnamefont{M.}~\bibnamefont{Lewenstein}},
  \bibinfo{journal}{Phys. Rev. A} \textbf{\bibinfo{volume}{61}},
  \bibinfo{pages}{062302} (\bibinfo{year}{2000}).

\bibitem[{\citenamefont{Wang and Sanders}(2003)}]{wang:entanglement}
\bibinfo{author}{\bibfnamefont{X.}~\bibnamefont{Wang}} \bibnamefont{and}
  \bibinfo{author}{\bibfnamefont{B.~C.} \bibnamefont{Sanders}},
  \bibinfo{journal}{Phys. Rev. A} \textbf{\bibinfo{volume}{68}},
  \bibinfo{pages}{014301} (\bibinfo{year}{2003}).

\bibitem[{\citenamefont{Bennett et~al.}(2003)\citenamefont{Bennett, Harrow,
  Leung, and Smolin}}]{bennett:capacities}
\bibinfo{author}{\bibfnamefont{C.}~\bibnamefont{Bennett}},
  \bibinfo{author}{\bibfnamefont{A.}~\bibnamefont{Harrow}},
  \bibinfo{author}{\bibfnamefont{D.}~\bibnamefont{Leung}}, \bibnamefont{and}
  \bibinfo{author}{\bibfnamefont{J.}~\bibnamefont{Smolin}},
  \bibinfo{journal}{Information Theory, IEEE Transactions on}
  \textbf{\bibinfo{volume}{49}}, \bibinfo{pages}{1895 } (\bibinfo{year}{2003}),
  ISSN \bibinfo{issn}{0018-9448}.

\bibitem{time_comment}
Note that the results of~\cite{lazyPaper} also apply for the state at other points in time, but are analogous
by substituting $\rhoSE \leftarrow \rhoSE(t)$.

\bibitem[{\citenamefont{Facchi and Pascazio}(2002)}]{facchi:subspaces}
\bibinfo{author}{\bibfnamefont{P.}~\bibnamefont{Facchi}} \bibnamefont{and}
  \bibinfo{author}{\bibfnamefont{S.}~\bibnamefont{Pascazio}},
  \bibinfo{journal}{Phys. Rev. Lett.} \textbf{\bibinfo{volume}{89}},
  \bibinfo{pages}{080401} (\bibinfo{year}{2002}).

\bibitem[{\citenamefont{Zanardi}(1999)}]{zanardi:symmetrizing}
\bibinfo{author}{\bibfnamefont{P.}~\bibnamefont{Zanardi}},
\bibinfo{journal}{Physics Letters A} \textbf{\bibinfo{volume}{258}},
  \bibinfo{pages}{77 } (\bibinfo{year}{1999}), ISSN \bibinfo{issn}{0375-9601}.

\bibitem[{\citenamefont{Viola et~al.}(1999)\citenamefont{Viola, Knill, and
  Lloyd}}]{viola:decoupling}
\bibinfo{author}{\bibfnamefont{L.}~\bibnamefont{Viola}},
  \bibinfo{author}{\bibfnamefont{E.}~\bibnamefont{Knill}}, \bibnamefont{and}
  \bibinfo{author}{\bibfnamefont{S.}~\bibnamefont{Lloyd}},
  \bibinfo{journal}{Phys. Rev. Lett.} \textbf{\bibinfo{volume}{82}},
  \bibinfo{pages}{2417} (\bibinfo{year}{1999}).

\bibitem[{\citenamefont{Rodr\'\i{}guez-Rosario}(May 2011)}]{cesarTalk}
\bibinfo{author}{\bibfnamefont{C.~A.} \bibnamefont{Rodr\'\i{}guez-Rosario}},
  \emph{\bibinfo{title}{Lazy states and quantum thermodynamics}}
  (\bibinfo{year}{May 2011}), \bibinfo{note}{talk given at CQT, Singapore}.

\bibitem[{\citenamefont{Ferraro et~al.}(2010)\citenamefont{Ferraro, Aolita,
  Cavalcanti, Cucchietti, and Ac\'\i{}n}}]{ferraro:discord}
\bibinfo{author}{\bibfnamefont{A.}~\bibnamefont{Ferraro}},
  \bibinfo{author}{\bibfnamefont{L.}~\bibnamefont{Aolita}},
  \bibinfo{author}{\bibfnamefont{D.}~\bibnamefont{Cavalcanti}},
  \bibinfo{author}{\bibfnamefont{F.~M.} \bibnamefont{Cucchietti}},
  \bibnamefont{and}
  \bibinfo{author}{\bibfnamefont{A.}~\bibnamefont{Ac\'\i{}n}},
  \bibinfo{journal}{Phys. Rev. A} \textbf{\bibinfo{volume}{81}},
  \bibinfo{pages}{052318} (\bibinfo{year}{2010}).

\bibitem{measure_comment}
Measure zero according to the Haar measure, choosing a pure state $U\proj{0}U\mdag$
by applying a randomly chosen unitary $U$.
  
  \bibitem[{\citenamefont{Kimura et~al.}(2007)\citenamefont{Kimura, Ohno, and
  Hayashi}}]{hayashi:purity}
\bibinfo{author}{\bibfnamefont{G.}~\bibnamefont{Kimura}},
  \bibinfo{author}{\bibfnamefont{H.}~\bibnamefont{Ohno}}, \bibnamefont{and}
  \bibinfo{author}{\bibfnamefont{H.}~\bibnamefont{Hayashi}},
  \bibinfo{journal}{Phys. Rev. A} \textbf{\bibinfo{volume}{76}},
  \bibinfo{pages}{042123} (\bibinfo{year}{2007}).

\bibitem[{\citenamefont{Gogolin}(2010)}]{gogolin:mthesis}
\bibinfo{author}{\bibfnamefont{C.}~\bibnamefont{Gogolin}}, Master's thesis
  (\bibinfo{year}{2010}), \bibinfo{note}{arXiv:1003.5058}.

\bibitem[{\citenamefont{Popescu et~al.}(2006)\citenamefont{Popescu, Short, and
  Winter}}]{popescu:entanglement}
\bibinfo{author}{\bibfnamefont{S.}~\bibnamefont{Popescu}},
  \bibinfo{author}{\bibfnamefont{A.}~\bibnamefont{Short}}, \bibnamefont{and}
  \bibinfo{author}{\bibfnamefont{A.}~\bibnamefont{Winter}},
  \bibinfo{journal}{Nature Physics} \textbf{\bibinfo{volume}{2}},
  \bibinfo{pages}{754} (\bibinfo{year}{2006}).

\bibitem[{\citenamefont{Dupuis}(2009)}]{dupuis:diss}
\bibinfo{author}{\bibfnamefont{F.}~\bibnamefont{Dupuis}}, Ph.D. thesis,
  \bibinfo{school}{Universit\'e de Montr\'eal} (\bibinfo{year}{2009}),
  \bibinfo{note}{arXiv:1004.1641}.

\bibitem[{\citenamefont{Dupuis et~al.}(2010)\citenamefont{Dupuis, Berta,
  Wullschleger, and Renner}}]{dupuis:decoupling}
\bibinfo{author}{\bibfnamefont{F.}~\bibnamefont{Dupuis}},
  \bibinfo{author}{\bibfnamefont{M.}~\bibnamefont{Berta}},
  \bibinfo{author}{\bibfnamefont{J.}~\bibnamefont{Wullschleger}},
  \bibnamefont{and} \bibinfo{author}{\bibfnamefont{R.}~\bibnamefont{Renner}}
  (\bibinfo{year}{2010}), \bibinfo{note}{arXiv:1012.6044v1}.

\bibitem[{\citenamefont{Nielsen and Chuang}(2000)}]{nielsen&chuang:qc}
\bibinfo{author}{\bibfnamefont{M.~A.} \bibnamefont{Nielsen}} \bibnamefont{and}
  \bibinfo{author}{\bibfnamefont{I.~L.} \bibnamefont{Chuang}},
  \emph{\bibinfo{title}{Quantum Computation and Quantum Information}}
  (\bibinfo{publisher}{Cambridge University Press}, \bibinfo{year}{2000}).

\bibitem[{\citenamefont{Helstrom}(1967)}]{helstrom:detection}
\bibinfo{author}{\bibfnamefont{C.~W.} \bibnamefont{Helstrom}},
  \bibinfo{journal}{Information and Control} \textbf{\bibinfo{volume}{10}},
  \bibinfo{pages}{254} (\bibinfo{year}{1967}).

\bibitem[{\citenamefont{Linden et~al.}(2009)\citenamefont{Linden, Popescu,
  Short, and Winter}}]{linden:evolution}
\bibinfo{author}{\bibfnamefont{N.}~\bibnamefont{Linden}},
  \bibinfo{author}{\bibfnamefont{S.}~\bibnamefont{Popescu}},
  \bibinfo{author}{\bibfnamefont{A.~J.} \bibnamefont{Short}}, \bibnamefont{and}
  \bibinfo{author}{\bibfnamefont{A.}~\bibnamefont{Winter}},
  \bibinfo{journal}{Phys. Rev. E} \textbf{\bibinfo{volume}{79}},
  \bibinfo{pages}{061103} (\bibinfo{year}{2009}).

\bibitem[{\citenamefont{Linden et~al.}(2010)\citenamefont{Linden, Popescu,
  Short, and Winter}}]{linden:fluctuations}
\bibinfo{author}{\bibfnamefont{N.}~\bibnamefont{Linden}},
  \bibinfo{author}{\bibfnamefont{S.}~\bibnamefont{Popescu}},
  \bibinfo{author}{\bibfnamefont{A.~J.} \bibnamefont{Short}}, \bibnamefont{and}
  \bibinfo{author}{\bibfnamefont{A.}~\bibnamefont{Winter}},
  \bibinfo{journal}{New Journal of Physics} \textbf{\bibinfo{volume}{12}},
  \bibinfo{pages}{055021} (\bibinfo{year}{2010}).

\bibitem[{\citenamefont{Renner}(2005)}]{renato:diss}
\bibinfo{author}{\bibfnamefont{R.}~\bibnamefont{Renner}}, Ph.D. thesis,
  \bibinfo{school}{ETH Zurich} (\bibinfo{year}{2005}),
  \bibinfo{note}{quant-ph/0512258}.

\bibitem[{\citenamefont{Boyd and Vandenberghe}(2004)}]{boyd:convex}
\bibinfo{author}{\bibfnamefont{S.}~\bibnamefont{Boyd}} \bibnamefont{and}
  \bibinfo{author}{\bibfnamefont{L.}~\bibnamefont{Vandenberghe}},
  \emph{\bibinfo{title}{Convex Optimization}} (\bibinfo{publisher}{Cambridge
  University Press}, \bibinfo{year}{2004}).


\end{thebibliography}

\newpage

\section{appendix}

This appendix is not necessary for the understanding of our work, and merely included for completeness sake.

{\bf A word on notation.}
Let $\cS(\hil_{A})$ denote the set of density operators on system $A$. For a density operator $\rho_{AB}\in\cS(\hil_{AB})$ the min-entropy of $A$ conditioned on $B$ is 
defined~\cite{renato:diss} as
\begin{align}\label{eq:hmin}
\hmin(A|B)_\rho := \sup_{\sigma_B \in\cS(\hil_B)}\sup\left\{\lambda\in\Real : 2^{-\lambda}\id_A\otimes\sigma_B\geq\rho_{AB}\right\}\ .
\end{align}
For a trivial system $B$ it simplifies to $\hmin (A)_\rho = -\log \lmax (\rho)$, where $\lmax$ denotes the largest eigenvalue. 

Let $\ket{\psi}_{AA'}:=\frac{1}{\sqrt{d_A}}\sum_{i=1}^{d_A}\ket{i}_A\otimes\ket{i}_{A'}$ denote the fully entangled state between $A$ and $A'$.
For a CPTPM $\cT_{A\ra B}$ (a completely positive and trace-preserving map) we define the \cj representation
\begin{align}
	\tau_{A'B}:=\left(\cI_{A'} \otimes \cT_{A\ra B}\right)\left(\proj{\psi}_{A'A}\right)
\end{align}
where $\cI_{A'}$ denotes the identity on $\End\left(\hil_{A'}\right)$. We now first establish an additional lemma that we will use in our proof.\smallskip

{\bf A result from quantum information theory.}
The following lemma -- a corollary of the Decoupling Theorem of \cite{dupuis:diss} -- gives a simple characterization of CPTPM's. If the min-entropy $\hmin (A'|B)_\tau$ of the \cj representation $\tau_{A'B}$ of a CPTPM as well as the dimension $d_A$ are large, then almost any input state $\rho_{A}$ will yield an output which is close to $\tau_{B}$.

\begin{lemma}\label{lem:decoupling}
Let $\rho_A\in\cS(\hil_{A})$ and let $\cT_{A\ra B}$ be a CPTPM with \cj representation $\tau_{A'B}$. Then,
\begin{align}
&\Pr_{U_A} \left\{\left\|\cT_{A\ra B}(U_A\rho_{A}U_A\mdag)-\tau_B\right\|_1 \geq 2^{-\half \hmin (A'|B)_\tau} + r\right\} 
\nn\\&\quad
\leq 2 e^{-d_A r^2/16}
\end{align}
where the probability is computed over the choice of $U$ from the Haar measure on the group of unitaries acting on $\hil_A$.
\end{lemma}

\begin{proof}
From \cite[Theorem 3.9]{dupuis:diss} with a trivial system $R$ we have for $\rho_A\in\cS(\hil_{A})$ that
\begin{align}
&\Pr_{U_A} \left\{\left\|\cT_{A\ra B}\left(U_A\rho_{A}U_A\mdag\right)-\tau_B\right\|_1 
\right.\nn\\&\qquad\left.
\geq 2^{-\half H_2 (A'|B)_\tau -\half H_2(A)_\rho} + r\right\} 
\nn\\&\quad
\leq 2 e^{-\frac{d_A r^2}{16K^2\left\|\rho_A\right\|_\infty}}
\end{align}
with $K = \max\left\{\left\|\cT(X)\right\|_1 : X\in\Herm(\hil_A), \left\|X\right\|_1\leq1\right\}$. The 2-entropy satisfies $H_2 (A'|B)_\tau \geq \hmin (A'|B)_\tau$ \cite[Lemma 2.3]{dupuis:diss}, and similarly $H_2(A)_\rho \geq \hmin(A)_\rho \geq 0$.
Since $\rho_A\in\cS(\hil_A)$ we have $\sqrt{\left\|\rho_A\right\|_\infty} \leq 1$. Any $X\in\Herm(\hil_A)$ can be written as $X = P_1 - P_2$ with $P_1,P_2\in\Herm(\hil_A)$, $P_1,P_2\geq0$. Since $\cT$ is trace-preserving and positive (i.e.\ maps positive operators to positive operators)
\begin{align}
\left\|\cT(X)\right\|_1
&\leq \left\|\cT(P_1)\right\|_1+\left\|\cT(P_2)\right\|_1 \nn\\
&= \tr\left[\cT(P_1)\right]+\tr\left[\cT(P_2)\right] \nn\\
&= \tr P_1 + \tr P_2 \nn\\
&= \left\|X\right\|_1\ ,
\end{align}
so 
\begin{align}\max\left\{\left\|\cT(X)\right\|_1 : X\in\Herm(\hil_A), \left\|X\right\|_1\leq1\right\} \leq 1\ .
\end{align}
Applying all these inequalities yields the assertion.
\end{proof}\smallskip

{\bf Proof of Lemma \ref{lem:closetoMixed}. }
\begin{proof}
Define a CPTPM by $\cT_{\sys\env\ra \sys}(\rhoSE)=\rhoS$, i.e.\ $\cT_{\sys\env\ra \sys}\equiv\tr\senv$. Then applying Lemma \ref{lem:decoupling} yields
\begin{align}\label{eqn:unitaryProb}
&\Pr_U \left\{
\left\|\tr\senv\left(U\rhoSE U\mdag\right) - \tau\ssys\right\|_1 
\geq 
2^{-\half \hmin (\sys'\env'|\sys)_\tau} + \beta\right\}
\nn\\&\quad 
\leq 2 e^{-d_Sd_E\beta^2/16}\ .
\end{align}
We have $\tau_{\sys'\env'\sys}=\tr\senv\proj{\psi}_{\sys\env\sys'\env'}$ so $\tau\ssys=\tr\senv\frac{\id_{\sys\env}}{\dS\dE}=\fmS$. The probability is computed over the choice of $U$ from the Haar measure on the group of unitaries on $\hil\ssys \otimes \hil\senv$. Applying a chain-rule for the min-entropy \cite[Lemma 3.1.10]{renato:diss} gives 
\begin{align}
\hmin (\sys'\env'|\sys)_\tau\geq\hmin (\sys'\env'\sys)_\psi-\log \dS\ .
\end{align}
It follows directly from the definition of the min-entropy that for a pure state $\sigma_{AB}$ we have $\hmin(A)_\sigma=\hmin(B)_\sigma$. Hence
\begin{align}
\hmin (\sys'\env'|\sys)_\tau
&\geq\hmin (\env)_\psi-\log \dS
\nn\\&=\hmin (\env)_{\frac{\id\senv}{\dE}}-\log \dS
\nn\\&=\log \dE-\log \dS\ .
\end{align}
Inserting this into (\ref{eqn:unitaryProb}) gives
\begin{align}
&\Pr_U \left\{
\left\|\tr\senv\left(U\rhoSE U\mdag\right) - \fmS\right\|_1 
\geq 
\sqrt{\frac{\dS}{\dE}}  + \beta\right\} 
\nn\\&\quad
\leq 2 e^{-\dS\dE\beta^2/16}
\end{align}
Chosing $\beta=\sqrt{\frac{\dS}{\dE}}$ we obtain (\ref{eq:prob}) 
with parameters (\ref{eq:parameters1}).
Alternatively, we choose $\beta=\dE^{-1/3}$ and obtain
\begin{align}
&\Pr_U \left\{
\left\|\tr\senv\left(U\rhoSE U\mdag\right) - \fmS\right\|_1 
\geq 
\sqrt{\frac{\dS}{\dE}}+\dE^{-1/3}\right\} 
\nn\\&\quad
\leq 2 e^{-\dS\dE^{1/3}/16}
\end{align}
Finally, we use $\sqrt{\frac{\dS}{\dE}}+\dE^{-1/3} \leq 2\frac{\sqrt{\dS}}{\sqrt[3]{\dE}}$ to find
\begin{align}
&\Pr_U \left\{
\left\|\tr\senv\left(U\rhoSE U\mdag\right) - \fmS\right\|_1 
\geq 
2\frac{\sqrt{\dS}}{\sqrt[3]{\dE}}
\right\}
\nn\\&\quad
\leq 2 e^{-\dS\dE^{1/3}/16}\ .
\end{align} 
This is (\ref{eq:prob}) with parameters (\ref{eq:parameters2}).

\end{proof}\smallskip

{\bf Pretty lazy for the purity.}
Here, we will extend our results about almost all states being ``pretty lazy'' to the case where we use the linear entropy or purity as a measure of decoherence instead of the von 
Neumann entropy. The purity of $\rho_S$ is simply given by $\tr(\rho_S^2)$. The rate of decoherence with respect to this measure is again
measured in terms of the time derivative
\begin{align}
	\left|\frac{d}{dt}	\tr(\rho_S(t)^2)\right|\ ,
\end{align}
where the condition for a zero rate of purity are exactly analogous. I.e., a particular state $\rho_{SE}$ is 
lazy with respect to purity having zero rate for any interaction Hamiltonian if and only if~\eqref{eq:lazyCondition} holds. For this measure of decoherence we have
\begin{align}\label{eq:mainResultPurity}
	\Pr_{\rho_{SE}}\left[\left|\frac{d}{dt}	\tr(\rho_S(t)^2)\right| \geq \|H_{\rm int}\|_{\infty} \chi\right] \leq \delta\ ,
\end{align}
where we may choose either parameters \eqref{eq:parameters1} or \eqref{eq:parameters2}.
With the parameters \eqref{eq:parameters1}, we only need \env to be larger than \emph{one} copy of $\sys$ in order to obtain a strong statement. Parameters \eqref{eq:parameters2} can also be applied in the case of an extremely small system $\log d_S \leq 2$. 

Let us now prove those claims.
Our argument is essentially analogous to the case of the von Neumann entropy: We already know that most states will be close
to maximally mixed on the system, which is itself a lazy state. It thus remains to show that states which are close to such a lazy state
are themselves pretty lazy.
We obtain a statement very similar to Lemma~\ref{lem:closeLazy}, however this time without an explicit dependence on $\dS$.

\begin{lemma}\label{lem:closeLazyPurity}
	Consider a Hamiltonian with interaction strength $\left\|H_{\rm int}\right\|_\infty$. 
	For any quantum state $\rhoSE$ on $\HSE$ such that its reduced state is $\chi$-close to fully mixed, i.e.,
	$\chi = \|\rho_\sys- \id_\sys/d_\sys\|_1$ where $\chi \leq 1/d_\sys$ with $d_\sys \geq 2$,  
	its purity rate is bounded by
	\begin{align}
		\left|\frac{d}{dt}	\tr(\rho_S(t)^2)\right| \leq \left\|H_{\rm int}\right\|_\infty \chi\ .
	\end{align}
\end{lemma}

\begin{proof}
A brief calculation \cite{lazyPaper} shows that similarly to \eqref{eq:vonNeumanRate} the rate of change of the purity is
\begin{align}
\frac{d}{dt}	\tr(\rho_S(t)^2) = i \tr\left(H_{\rm int}\left[\rho_\sys(t) \otimes \id_\env,\rhoSE(t) \right]\right)\ .
\end{align}
Following the same procedure as in the derivation of \eqref{eq:stepLastBound} we find
\begin{align}
&\left|\frac{d}{dt}	\tr(\rho_S(t)^2)\right| 
\leq \left\|H_{\rm int}\right\|_\infty\left\|\left[\rho_S(t)\otimes\id_E,\rho_{SE}\right]\right\|_1
\\&= \left\|H_{\rm int}\right\|_\infty\left\|\left[\left(\rho_S(t)-\fmS\right)\otimes\id_E,\rho_{SE}\right]\right\|_1
\\&\leq 2\left\|H_{\rm int}\right\|_\infty\left\|\left(\rho_S(t)-\fmS\right)\otimes\id_E\right\|_\infty\left\|\rho_{SE}\right\|_1
\\&= 2\left\|H_{\rm int}\right\|_\infty\left\|\rho_S(t)-\fmS\right\|_\infty\ .
\end{align}
Now let $\left\{p_i\right\}_{i=1}^{d_S}$ denote the eigenvalues of $\rho_S(t)$. Since $\sum_{i=1}^{\dS}|p_i-1/\dS| = \chi$ it is clear that
\begin{align}
\left\|\rho_S(t)-\fmS\right\|_\infty
&= \max_{i=1}^{\dS}\left|p_i-1/\dS\right|
\\&\leq \half\sum_{i=1}^{\dS}\left|p_i-1/\dS\right|
\\&= \half\left\|\rho_S(t)-\fmS\right\|_1
\\&= \half\chi
\end{align}
and hence the assertion.
\end{proof}

The statement about low purity rates \eqref{eq:mainResultPurity} then follows through direct combination of Lemma~\ref{lem:closetoMixed} and Lemma~\ref{lem:closeLazyPurity}.\smallskip

{\bf Upper bound on the entropy rate for arbitrary states and Hamiltonians.}
\begin{lemma}\label{lem:upperBoundVNRate}
For any bipartite Hamiltonian $\HSE$ with interaction strength $\|H_{\rm int}\|_{\infty}$ we have
\begin{align}
	\left|\frac{d\ent(\sys)}{dt}\right| \leq 4 \|H_{\rm int}\|_{\infty} \log d_{\sys}\ . 
\end{align}
\end{lemma}

This bound holds for any state $\rhoSE$, pure or mixed, the joint system may be in. 

\begin{proof}
Let the state of $\sys\env$ be given by $\rhoSE$. Since we did not impose any restrictions on the Hamiltonian whatsoever, we can formally extend the environment with a purifying system $P$ and extend the Hamiltonian to $H_{SEP}=H_{SE}\otimes\id_P$. The interactive part of the Hamiltonian $H_{int}$ gets an additional factor $\id_P$ so that the quantities $\left\|H_{\rm int}\right\|_\infty$ and $\Delta(H_{\rm int})$ are invariant under this extension.

Let $\rho_{SEP}=\proj{\mu}_{SEP}$. 
Then by use of \eqref{eq:vonNeumanRate} and \eqref{normineq}
\begin{align}
\left|\frac{d\ent(\sys)}{dt}\right| 
\leq \left\|H_{\rm int}\right\|_\infty \left\|\left[\log(\rho_S)\otimes\id_{EP},\proj{\mu}_{SEP}\right]\right\|_1\ .
\end{align}
Now let $\ket{\nu}_{S\tilde{P}}$ denote a purification of $\rho_S$. Since both $\ket{\nu}_{S\tilde{P}}$ and $\ket{\mu}_{SEP}$ are purifications of $\rho_S$, there is an isometry $V_{\tilde{P}\ra EP}$ with $V_{\tilde{P}\ra EP}\ket{\nu}_{S\tilde{P}}=\ket{\mu}_{SEP}$. Hence,
\begin{align}
&\left|\frac{d\ent(\sys)}{dt}\right|  
\nn\\& \leq \left\|H_{\rm int}\right\|_\infty \left\|\left[\log(\rho_S)\otimes\left(V_{\tilde{P}\ra EP}\id_{\tilde{P}}V_{\tilde{P}\ra EP}\mdag\right),
\right.\right.\nn\\&\qquad\qquad\qquad\left.\left.
V_{\tilde{P}\ra EP}\proj{\nu}_{S\tilde{P}}V_{\tilde{P}\ra EP}\mdag\right]\right\|_1
\nn\\& = \left\|H_{\rm int}\right\|_\infty \left\|V_{\tilde{P}\ra EP}\left[\log(\rho_S)\otimes\id_{\tilde{P}},\proj{\nu}_{S\tilde{P}}\right]V_{\tilde{P}\ra EP}\mdag\right\|_1
\nn\\& = \left\|H_{\rm int}\right\|_\infty \left\|\left[\log(\rho_S)\otimes\id_{\tilde{P}},\proj{\nu}_{S\tilde{P}}\right]\right\|_1\ .
\end{align}
The commutator may therefore be calculated for an arbitrary purification $\ket{\nu}_{S\tilde{P}}$ of $\rho_S(t)$. The operator $\iu\left[\log(\rho_S)\otimes\id_{\tilde{P}},\proj{\nu}_{S\tilde{P}}\right]$ is Hermitian and has vanishing trace, so its eigenvalues are real and sum up to zero. The operator $\Pi_{S\tilde{P}}$ which is the projection onto the eigenstates with positive eigenvalues therefore allows to write
\begin{align}
&\left\|\iu\left[\log(\rho_S)\otimes\id_{\tilde{P}},\proj{\nu}_{S\tilde{P}}\right]\right\|_1
\nn\\&=2\tr\left\{\Pi_{SP}\iu\left[\log(\rho_S)\otimes\id_{\tilde{P}},\proj{\nu}_{S\tilde{P}}\right]\Pi_{SP}\right\}
\nn\\&=2\iu\tr\left\{\left[\Pi,\log(\rho)\otimes\id\right]\proj{\nu}\right\}
\nn\\&=2\iu\bra{\nu}\left[\Pi,\log\rho\otimes\id\right]\ket{\nu}
\nn\\&\leq4\left|\bra{\nu}\Pi\left(\log\rho\otimes\id\right)\ket{\nu}\right|
\nn\\&\leq4\sqrt{\bra{\nu}\Pi\Pi\mdag\ket{\nu}}\sqrt{\bra{\nu}\left(\log\rho\otimes\id\right)\mdag\left(\log\rho\otimes\id\right)\ket{\nu}}
\nn\\&\leq4\sqrt{\bra{\nu}\left(\log\rho\otimes\id\right)^2\ket{\nu}}
\nn\\&=4\sqrt{\sum_{i=1}^{d_S}p_i\left(\log p_i\right)^2}
\nn\\&\leq4\log d_S\ .
\end{align}
The second inequality is due to an application of Cauchy-Schwarz, the last one can be proved by use of a Lagrange multiplier.
\end{proof}

\end{document}